\documentclass[12pt,a4paper]{article}
\usepackage{amsmath,amssymb,amsthm} 
\newtheorem{lemma}{Lemma}
\newtheorem{proposition}{Proposition}
\newtheorem{theorem}{Theorem}

\theoremstyle{definition}

\theoremstyle{remark}
\newtheorem{remark}{Remark}
\newcommand{\beq}{\begin{eqnarray}}
\newcommand{\eeq}{\end{eqnarray}}
\newcommand{\beqnn}{\begin{eqnarray*}}
\newcommand{\eeqnn}{\end{eqnarray*}}
\newcommand{\rd}{\partial}

\newcommand{\tp}[1]{\:{}^{\mathrm{t}}#1}

\newcommand{\ZZ}{\mathbb{Z}}
\newcommand{\bst}{\boldsymbol{t}}
\newcommand{\bsx}{\boldsymbol{x}}
\newcommand{\bsy}{\boldsymbol{y}}
\newcommand{\bszero}{\boldsymbol{0}}
\newcommand{\calP}{\mathcal{P}}
\newcommand{\GL}{\mathrm{GL}}
\newcommand{\gl}{\mathrm{gl}}
\newcommand{\LHS}{\mathrm{LHS}}
\newcommand{\RHS}{\mathrm{RHS}}

\begin{document}

\title{Generalized Ablowitz-Ladik hierarchy in topological string theory}
\author{Kanehisa Takasaki
\thanks{E-mail: takasaki@math.h.kyoto-u.ac.jp}\\
{\normalsize Graduate School of Human and Environmental Studies}\\ 
{\normalsize Kyoto University}\\
{\normalsize Sakyo, Kyoto 606-8501, Japan}}
\date{}
\maketitle 

\begin{abstract}
This paper addresses the issue of integrable structures  
in topological string theory on generalized conifolds.   
Open string amplitudes of this theory can be expressed 
as the matrix elements of an operator on the Fock space 
of 2D charged free fermion fields.  The generating function 
of these amplitudes with respect to the product of 
two independent Schur functions becomes a tau function 
of the 2D Toda hierarchy.  The associated Lax operators 
turn out to have a particular factorized form.  This factorized form 
of the Lax operators characterizes a generalization of 
the Ablowitz-Ladik hierarchy embedded in the 2D Toda hierarchy.  
The generalized Ablowitz-Ladik hierarchy is thus identified 
as a fundamental integrable structure of topological string theory 
on the generalized conifolds.  
\end{abstract}

\begin{flushleft}
Pacs Numbers: 02.30.Ik, 11.25.Hf
\end{flushleft}


\section{Introduction}

This paper presents an extension of our previous work \cite{Takasaki13} 
on the integrable structure of a modified ``melting crystal model''. 
The partition function of the previous model is a generating function 
of (equivariant) local Gromov-Witten invariants of the resolved conifold 
\cite{BP08}. The relevance of the Ablowitz-Ladik hierarchy 
\cite{AL75,Vekslerchik97}  (equivalently, 
the relativistic Toda hierarchy \cite{Ruijsenaars90,KMZ96,Suris97}) 
to this model was conjectured by Brini from an order-by-order analysis 
of the genus expansion of the partition function \cite{Brini10}.   
We deformed the modified melting crystal model 
by two sets of external potentials 
with coupling constants $\bst = (t_1,t_2,\ldots)$ 
and $\bar{\bst}= (\bar{t}_1,\bar{t}_2,\ldots)$, 
and employed the method \cite{NT07,NT08} developed 
for the genuine melting crystal model \cite{ORV03,MNTT04} 
to show that the deformed partition function is essentially 
a tau function $\tau(s,\bst,\bar{\bst})$ of the 2D Toda hierarchy 
\cite{UT84,TT95}.  We further found that the associated Lax operators 
of the Toda hierarchy have a particular factorized form, 
which coincides with the condition  that characterizes 
the Ablowitz-Ladik hierarchy as a reduction of the Toda hierarchy \cite{BCR11}. 
Thus the solution of the Toda hierarchy in question turns out 
to be a solution of the Ablowitz-Ladik hierarchy.  

In the perspective of topological string theory, the tau function 
of the modified melting crystal model is a generating function 
of open string amplitudes on the resolved conifold.    
One can generalize it to toric Calabi-Yau threefolds 
whose toric diagrams are obtained by ``triangulation of a strip''.  
For such manifolds, the method of topological vertex \cite{AKMV03} 
for calculating open string amplitudes simplifies drastically \cite{IKP04}.  
Namely, the amplitudes can be expressed as the matrix elements 
$\langle\lambda|g|\mu\rangle$ of an operator $g$ 
on the Fock space of 2D charged free fermion fields, 
where $\lambda$ and $\mu$ are partitions assigned 
to the leftmost and rightmost external edges of the web diagram, 
and $\langle\lambda|$ and $|\mu\rangle$ are normalized excited states 
labelled by these partitions \cite{EK03,Nagao09,Sulkowski09}.   
The generating function of the matrix elements 
$\langle\lambda,s|g|\mu,s\rangle$ 
(extended to the charge-$s$ sectors of the Fock space) 
with respect to the product $s_\lambda(\bsx)s_\mu(\bsy)$ 
of Schur functions of two independent sets of variables $\bsx,\bsy$ 
becomes a  tau function $\tau(s,\bst,\bar{\bst})$ 
of the Toda hierarchy by a well known transformation of variables 
from $\bsx,\bsy$ to $\bst,\bar{\bst}$ \cite{Takasaki13b}. 

Among these generalizations to toric Calabi-Yau threefolds, 
we now consider the case where the web diagram consists 
of a zig-zag spine of $2N-1$ internal edges and 
$2N+2$ external edges emanating from the spine 
(see Fig. \ref{web-diagram}).   
The resolved conifold amounts to the case of $N = 1$.   
\begin{figure}
\begin{center}
\unitlength=1pt
\begin{picture}(310,200)
\put(10,50){\makebox{$\tp{\lambda}$}}
\put(0,40){\line(1,0){40}}
\put(40,40){\line(0,-1){40}}
\put(40,40){\line(1,1){30}}
\put(55,45){\makebox{$Q_1$}}
\put(70,70){\line(0,1){40}}
\put(70,70){\line(1,0){40}}
\put(85,80){\makebox{$Q_2$}}
\put(110,70){\line(0,-1){40}}
\put(110,70){\line(1,1){30}}
\put(125,75){\makebox{$Q_3$}}
\put(140,100){\line(1,0){35}}
\put(140,100){\line(0,1){40}}
\multiput(180,105)(5,5){5}{\circle*{1}}
\put(205,130){\line(1,0){35}}
\put(240,130){\line(0,-1){40}}
\put(240,130){\line(1,1){30}}
\put(255,135){\makebox{$Q_{2N-1}$}}
\put(270,160){\line(0,1){40}}
\put(270,160){\line(1,0){40}}
\put(290,170){\makebox{$\mu$}}
\end{picture}
\end{center}
\caption{Web diagram of generalized conifold}
\label{web-diagram}
\end{figure}
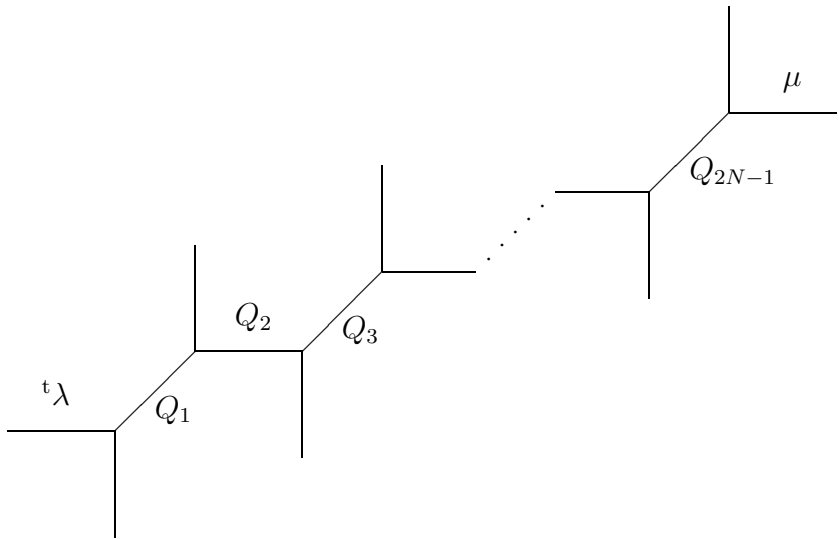
For this ``generalized conifold'', we shall show that 
the associated Lax operators of the Toda hierarchy 
have the factorized form 
\beq
  L = Be^{(1-N)\rd_s}C^{-1},\quad 
  \bar{L}^{-1} = DCe^{(N-1)\rd_s}B^{-1}, 
\label{LLbar-BCop}
\eeq
where $e^{\rd_s}$ ($\rd_s = \rd/\rd s$) denotes 
the shift operator that acts on functions $\psi(s)$ of  $s \in \ZZ$ 
as $e^{\rd_s}\psi(s) = \psi(s+1)$,  
$B$ and $C$ are finite order difference operators of the form 
\beq
\begin{aligned}
  B &= e^{N\rd_s} + b_1(s)e^{(N-1)\rd_s} + \cdots + b_N(s),\\
  C &= 1 + c_1(s)e^{-\rd_s} + \cdots + c_N(s)e^{-N\rd_s}, 
\end{aligned}
\label{general-BCop}
\eeq
and $D$ is a constant.   The factorized form (\ref{LLbar-BCop}) 
is preserved by the flows of the Toda hierarchy, 
hence defines a generalization of the Ablowitz-Ladik hierarchy 
as a subsystem of the Toda hierarchy.    It is this generalized 
Ablowitz-Ladik hierarchy that we identify as a fundamental 
integrable structure of topological string theory 
on the $N$-th generalized conifold.   

To derive this result, we adopt the same strategy 
as in the previous work \cite{Takasaki13}.  Namely, 
we translate the fermionic construction of the tau function 
to the language of $\ZZ\times\ZZ$ matrices.    The operator $g$ 
thus corresponds to an infinite matrix $U$.  The associated solution 
of the Toda hierarchy can be captured by a factorization problem 
of infinite matrices \cite{Takasaki84,NTT95,Takasaki95}. 
Actually, infinite matrices arising therein are avatars of 
difference operators in the Lax formalism of the Toda hierarchy.  
The factors in the factorization problem amount to 
the dressing operators $W,\bar{W}$, which in turn 
determine the Lax operators $L,\bar{L}$.   
We can find an explicit form of these matrices at least 
at the initial time $\bst = \bar{\bst} =\bszero$.  
The initial values of $L$ and $\bar{L}^{-1}$ turns out 
to have the factorized form (\ref{LLbar-BCop}). 
This is enough to conclude that $L$ and $\bar{L}^{-1}$ 
are similarly factorized at all times, because 
the factorized form (\ref{LLbar-BCop}) is preserved 
by the flows of the Toda hierarchy.  

This paper is organized as follows. Section 2 reviews 
the fermionic expression of open string amplitudes 
and the associated tau function.   The operator $g$ 
shows up here as a product of several building blocks. 
Section 3 shows  a dictionary that translates 
the fermionic setting to the language of infinite matrices.    
The matrix $U$ is derived from $g$ by this dictionary. 
Section 4 introduces the matrix factorization problem 
that determines the dressing operators.  
The initial values of the dressing operators 
are explicitly calculated here.   This result 
is used in Section 5 to calculate the initial values 
of the Lax operators.   The initial values are shown 
to have the factorized form (\ref{LLbar-BCop}).  
Section 6 is devoted to the generalized 
Ablowitz-Ladik hierarchy.  Fundamental properties, 
such as the consistency of this system, are proved here.  
As a corollary, the solution of the Toda hierarchy 
obtained from $U$ turns out to be a solution of 
the generalized Ablowitz-Ladik hierarchy.   
Section 7 presents our conclusion.

\section{Tau function in fermionic language}

Let $\psi_n$ and $\psi^*_n$, $n \in \ZZ$, be the Fourier modes 
of 2D charged free fermion fields \cite{MJD-book}
\beqnn
  \psi(z) = \sum_{n\in\ZZ}\psi_nz^{-n-1},\quad 
  \psi^*(z) = \sum_{n\in\ZZ}\psi^*_nz^{-n}. 
\eeqnn
They satisfy the anti-commutation relations 
\beqnn
   \psi_m\psi^*_n + \psi^*_n\psi_n = \delta_{m+n,0},\quad 
   \psi_m\psi_n + \psi_n\psi_m = 0, \quad 
  \psi^*_m\psi^*_n + \psi^*_n\psi^*_m = 0. 
\eeqnn
The Fock and dual Fock space can be decomposed 
into the eigenspaces of the $U(1)$-charge operator 
\beqnn
  J_0 = \sum_{n\in\ZZ}{:}\psi_{-n}\psi^*_n{:}\,. 
\eeqnn
Let $\langle s|$ and $|s\rangle$, $s \in \ZZ$, denote 
the ground states in the charge-$s$ sector: 
\beqnn
  \langle s| = \langle-\infty|\cdots\psi^*_{s-1}\psi^*_s,\quad 
  |s\rangle = \psi_{-s}\psi_{-s+1}\cdots|-\infty\rangle
\eeqnn
Normalized excited states are labelled by partitions 
$\lambda = (\lambda_1,\lambda_2,\ldots)$, 
$\lambda_1 \ge \lambda_2 \ge \cdots \ge 0$, 
as $\langle \lambda,s|$ and $|\lambda,s\rangle$:
\beqnn
  \langle\lambda,s| 
  = \langle-\infty|\cdots\psi^*_{\lambda_2+s-1}\psi^*_{\lambda_1+s},\quad 
  |\lambda,s\rangle 
  = \psi_{-\lambda_1-s}\psi_{-\lambda_2-s+1}\cdots|-\infty\rangle.
\eeqnn

Building blocks of the fermionic expression of open string amplitudes 
are the special fermion bilinears 
\beqnn
\begin{gathered}
  J_k = \sum_{n\in\ZZ}{:}\psi_{-n}\psi^*_{n+k}{:},\quad k \in \ZZ,\\
  L_0 = \sum_{n\in\ZZ}n{:}\psi_{-n}\psi^*_n{:},\quad
  W_0 = \sum_{n\in\ZZ}n^2{:}\psi_{-n}\psi^*_n{:} 
\end{gathered}
\eeqnn
and the vertex operators \cite{OR01,BY08}
\beqnn
  \Gamma_{\pm}(z) 
    = \exp\left(\sum_{k=1}^\infty\frac{z^k}{k}J_{\pm k}\right),\quad
  \Gamma'_{\pm}(z) 
    = \exp\left(-\sum_{k=1}^\infty\frac{(-z)^k}{k}J_{\pm k}\right). 
\eeqnn
It is convenient to define the multi-variate vertex operators 
$\Gamma_{\pm}(\bsx)$ and $\Gamma'_{\pm}(\bsx)$, 
$\bsx = (x_1,x_2,\ldots)$, as well: 
\beqnn
  \Gamma_{\pm}(\bsx) = \prod_{i=1}^\infty\Gamma_{\pm}(x_i),\quad 
  \Gamma'_{\pm}(\bsx) = \prod_{i=1}^\infty\Gamma'_{\pm}(x_i). 
\eeqnn

These operators are ``neutral'' or ``charge-preserving'', 
i.e., do not mix different charge sectors.  $L_0$ and $W_0$ 
are ``diagonal'' and have the $s$-dependent matrix elements 
\begin{align}
  \langle\lambda,s|L_0|\mu,s\rangle   
    &= \delta_{\lambda\mu}
       \left(|\lambda| + \frac{s(s+1)}{2}\right), 
  \label{<ls|L0|ms>}\\
  \langle\lambda,s|W_0|\mu,s\rangle
    &= \delta_{\lambda\mu}
       \left(\kappa(\lambda) + (2s+1)|\lambda| 
        + \frac{s(s+1)(2s+1)}{6}\right), 
  \label{<ls|W0|ms>}
\end{align}
where $\kappa(\lambda)$ denotes the quadratic Casimir invariant 
\beqnn
  \kappa(\lambda) 
  = \sum_{i\ge 1}\lambda_i(\lambda_i-2i+1) 
  = \sum_{i\ge 1}\left((\lambda_i - i + 1/2)^2 - (-i + 1/2)^2\right) 
\eeqnn
of the partition $\lambda = (\lambda_1,\lambda_2,\ldots)$. 
The matrix elements of $\Gamma_{\pm}(\bsx)$ and 
$\Gamma'_{\pm}(\bsx)$ are independent of 
the charge sector, and become skew Schur functions 
\cite{MJD-book,Macdonald-book}:  
\begin{gather}
  \langle\lambda,s|\Gamma_{-}(\bsx)|\mu,s\rangle
  = \langle\mu,s|\Gamma_{+}(\bsx)|\lambda,s\rangle
  = s_{\lambda/\mu}(\bsx),
  \label{<ls|Gamma|ms>}\\
  \langle\lambda,s|\Gamma'_{-}(\bsx)|\mu,s\rangle 
  = \langle\mu,s|\Gamma'_{+}(\bsx)|\lambda,s\rangle
  = s_{\tp{\lambda}/\tp{\mu}}(\bsx), 
  \label{<ls|Gamma'|ms>}
\end{gather}
where $\tp{\lambda}$ and $\tp{\mu}$ denote 
the conjugate (or transposed) partitions 
of $\lambda$ and $\mu$.  

$J_k$'s play a fundamental role 
in the fermionic expression of tau functions 
of the Toda hierarchy as well \cite{Takebe91}.  
General tau functions can be expressed as 
\beq
  \tau(s,\bst,\bar{\bst}) 
  = \langle s|\exp\left(\sum_{k=1}^\infty t_kJ_k\right)
    g \exp\left(- \sum_{k=1}^\infty\bar{t}_kJ_{-k}\right)
    |s\rangle, 
\label{general-tau}
\eeq
where $g$ is an element of the Clifford group $\widehat{\GL}(\infty)$ 
\cite{MJD-book} generated by  the exponentials $e^{\hat{X}}$ 
of neutral fermion bilinears  
\beqnn
  \hat{X} = \sum_{m,n\in\ZZ}x_{mn}{:}\psi_{-m}\psi^*_n{:}\,. 
\eeqnn
The tau function of the modified melting crystal model 
\cite{Takasaki13} is obtained from the special operator 
\beq
  g = q^{W_0/2}\Gamma_{-}(q^{-\rho})\Gamma_{+}(q^{-\rho})Q^{L_0}
       \Gamma'_{-}(q^{-\rho})\Gamma'_{+}(q^{-\rho})q^{-W_0/2},  
\label{conifold-g}
\eeq
where $q$ and $Q$ are positive constants,  $q$ is restricted 
to the range $0 < q < 1$,  and $\Gamma_{\pm}(q^{-\rho})$ 
and $\Gamma'_{\pm}(q^{-\rho})$ are the specialization 
of $\Gamma_{\pm}(\bsx)$ and $\Gamma'_{\pm}(\bsx)$ to 
\beqnn
  q^{-\rho} = \left(q^{1/2},q^{3/2},\ldots,q^{k+1/2},\ldots\right). 
\eeqnn

Topological string theory on the $N$-th generalized conifold 
is related to the following generalization of (\ref{conifold-g}) 
\cite{Nagao09,Sulkowski09}: 
\begin{align}
  g &= q^{W_0/2}\Gamma_{-}(q^{-\rho})\Gamma_{+}(q^{-\rho})Q_1^{L_0}
       \Gamma'_{-}(q^{-\rho})\Gamma'_{+}(q^{-\rho})Q_2^{L_0}\nonumber\\
    &\quad\mbox{}\times 
         \Gamma_{-}(q^{-\rho})\Gamma_{+}(q^{-\rho})Q_3^{L_0}
         \Gamma'_{-}(q^{-\rho})\Gamma'_{+}(q^{-\rho})Q_4^{L_0}\nonumber\\
    &\quad\mbox{}\times\cdots\nonumber\\
    &\quad\mbox{}\times 
         \Gamma_{-}(q^{-\rho})\Gamma_{+}(q^{-\rho})Q_{2N-1}^{L_0}
         \Gamma'_{-}(q^{-\rho})\Gamma'_{+}(q^{-\rho})q^{-W_0/2}. 
\label{gencon-g}
\end{align}
$Q_1,\ldots,Q_{2N-1}$ are ``K\"ahler parameters'' assigned 
to the internal edges of the web diagram (see Fig. \ref{web-diagram}). 
The matrix elements $\langle\lambda|g|\mu\rangle$ of $g$ 
in the Fock space give a compact expression 
of the open string amplitudes $Z_{\tp{\lambda}\emptyset\cdots\emptyset\mu}$ 
in the setting where the leftmost and rightmost external edges 
are assigned with the non-trivial partitions $\tp{\lambda},\mu$ 
whereas the other external edges are given the zero partition 
$\emptyset$
\footnote{As we show in Appendix A, the matrix elements 
$\langle\lambda|g|\mu\rangle$ are slightly different 
from the true amplitudes 
$Z_{\tp{\lambda}\emptyset\cdots\emptyset\mu}$.  
The difference, however, is rather immaterial and does not 
affect our consideration on the integrable structure.}.  

The logarithm of $Z_{\tp{\lambda}\emptyset\cdots\emptyset\mu}$ 
is a generating functions of all genus 
local Gromov-Witten invariants 
with respect to the parameters $q,Q_1,\ldots,Q_{2N-1}$.  
The partitions specify the topological type 
(i.e., winding numbers) of boundaries of the string world sheet. 
One can assign non-trivial partitions $\beta_1,\ldots,\beta_{2N}$ 
to the intermediate external edges to define more general 
amplitudes $Z_{\tp{\lambda}\beta_1\ldots\beta_{2N}\mu}$ as well. 
We refer a precise definition of these notions 
to Mari\~{n}o's book \cite{Marino-book}.  
The fermionic expression of these amplitudes 
can be derived from a combinatorial expression 
obtained by the method of topological vertex
(see Appendix A).  

The matrix elements $\langle\lambda|g|\mu\rangle$ 
and their extensions $\langle\lambda,s|g|\mu,s\rangle$ 
to the charge-$s$ sectors can be assembled 
to the generating functions 
\beq
   \sum_{\lambda,\mu\in\calP}s_{\lambda}(\bsx) 
     \langle\lambda,s|g|\mu,s\rangle s_\mu(\bsy), 
\label{gen-function}
\eeq
where $\calP$ is the set of all partitions, and 
$s_\lambda(\bsx)$ and $s_\mu(\bsy)$ are the Schur functions 
of multi-variables $\bsx = (x_1,x_2,\ldots)$ 
and $\bsy = (y_1,y_2,\ldots)$. These generating functions 
become the tau function (\ref{general-tau})  
by changing variables from $(\bsx,\bsy)$ 
to $(\bst,\bar{\bst})$ as 
\beq
  t_k = \frac{1}{k}\sum_{i\ge 1} x_i^k, \quad 
  \bar{t}_k = - \frac{1}{k}\sum_{i\ge 1} y_i^k.  
\eeq

\section{Translation to infinite matrices}

The foregoing fermionic setting can be translated 
to the language of infinite matrices.    In particular, we obtain 
a matrix that represents the element (\ref{gencon-g}) 
of $\widehat{\GL}(\infty)$.   We shall use this matrix 
to characterize the associated solution of the Toda hierarchy 
in the Lax formalism.  

It is well known that neutral fermion bilinears are 
in one-to-one correspondence with 
$\ZZ\times\ZZ$ matrices \cite{MJD-book}: 
\beqnn
  \hat{X} \;\longleftrightarrow\;
  X = \sum_{m,n\in\ZZ}x_{mn}E_{mn}, \quad 
  E_{mn} = (\delta_{im}\delta_{jn})_{i,j\in\ZZ}. 
\eeqnn
Apart from c-number corrections, they have 
the same Lie algebraic structure, namely, 
\beqnn
  [\hat{X},\hat{Y}] = \widehat{[X,Y]} + c(X,Y), 
\eeqnn
where $c(X,Y)$ is a c-number cocycle of $\gl(\infty)$.   
Thus fermion bilinears form a central extension $\widehat{\gl}(\infty)$ 
of $\gl(\infty)$. 

Actually, the matrix representation gives us more freedom, 
because the set of infinite matrices is equipped with 
multiplication as well as Lie brackets.   For example, 
matrix representation of $L_0, M_0, J_k$ 
(let us use the same notations as fermion bilinears) 
can be written as 
\beq
  L_0 = \Delta,\quad 
  W_0 = \Delta^2,\quad 
  J_k = \Lambda^k, 
\label{LWJ-matrix}
\eeq
where 
\beqnn
  \Delta = \sum_{n\in\ZZ}nE_{nn},\quad 
  \Lambda = \sum_{n\in\ZZ}E_{n,n+1}. 
\eeqnn
Note that $\Lambda$ and $\Delta$ amount to 
the shift operator $e^{\rd_s}$ and 
the multiplication operator $s$ on the 1D lattice $\ZZ$: 
\beqnn
  \Lambda \;\longleftrightarrow\; e^{\rd_s},\quad 
  \Delta \;\longleftrightarrow\; s. 
\eeqnn

The matrix representation of fermion bilinears can be lifted up 
to $\widehat{\GL}(\infty)$  \cite{MJD-book}.   
By definition, the matrix representation $A = (a_{ij})$ 
of $g \in \widehat{\GL}(\infty)$ is determined by 
the Bogoliubov transformation on the linear span of $\psi_n$'s: 
\beqnn
   g\psi_ng^{-1} = \sum_{m\in\ZZ}\psi_ma_{mn}. 
\eeqnn
If $g$ is the exponential $e^{\hat{X}}$ 
of a neutral fermion bilinear $\hat{X}$, 
$A$ is the exponential $e^X$ of the matrix $X$.  
We can thereby find, from (\ref{LWJ-matrix}),  
an explicit form of the matrix representation 
of building blocks in (\ref{gencon-g}): 
\begin{itemize}
\item[(i)]
The matrix representation of $\Gamma_{\pm}(z)$ 
and $\Gamma'_{\pm}(z)$, denoted by the same notations, reads 
\beq
\begin{gathered}
  \Gamma_{\pm}(z) 
  = \exp\left(\sum_{k=1}^\infty\frac{z^k}{k}\Lambda^{\pm k}\right) 
  = (1 - z\Lambda^{\pm 1})^{-1},\\
  \Gamma'_{\pm}(z) 
  = \exp\left(- \sum_{k=1}^\infty\frac{(-z)^k}{k}\Lambda^{\pm k}\right) 
  = 1 + z\Lambda^{\pm}. 
\end{gathered}
\label{Gamma(z)-matrix}
\eeq
Consequently, $\Gamma_{\pm}(q^{-\rho})$ and $\Gamma'_{\pm}(q^{-\rho})$ 
become infinite products of matrices: 
\beq
\begin{aligned}
  \Gamma_{\pm}(q^{-\rho}) 
  &= \prod_{i=1}^\infty (1 - q^{i-1/2}\Lambda^{\pm 1})^{-1},\\
  \Gamma'_{\pm}(q^{-\rho}) 
  &= \prod_{i=1}^\infty (1 + q^{i-1/2}\Lambda^{\pm 1}). 
\end{aligned}
\label{Gamma(q)-matrix}
\eeq
These infinite products may be thought of as matrix-valued 
``quantum dilogarithm'' \cite{FV93,FK93}.   
\item[(ii)]
$q^{W_0/2}$ and $Q^{L_0}$ become the diagonal matrices 
\beq
  q^{\Delta^2/2} = \sum_{n\in\ZZ}q^{n^2/2}E_{nn},\quad 
  Q^\Delta = \sum_{n\in\ZZ}Q^nE_{nn}. 
\eeq
\end{itemize}
Thus the following matrix representation of (\ref{gencon-g}) is obtained: 
\begin{align}
  U &= q^{\Delta^2/2}\Gamma_{-}(q^{-\rho})\Gamma_{+}(q^{-\rho})Q_1^\Delta
         \Gamma'_{-}(q^{-\rho})\Gamma'_{+}(q^{-\rho})Q_2^\Delta\nonumber\\
    &\quad\mbox{}\times 
         \Gamma_{-}(q^{-\rho})\Gamma_{+}(q^{-\rho})Q_3^\Delta
         \Gamma'_{-}(q^{-\rho})\Gamma'_{+}(q^{-\rho})Q_4^\Delta\nonumber\\
    &\quad\mbox{}\times \cdots\nonumber\\
    &\quad\mbox{}\times 
         \Gamma_{-}(q^{-\rho})\Gamma_{+}(q^{-\rho})Q_{2N-1}^\Delta
         \Gamma'_{-}(q^{-\rho})\Gamma'_{+}(q^{-\rho})q^{-\Delta^2/2}. 
\label{gencon-U}
\end{align}

Let us make several technical remarks: 

\begin{remark}
$Q^\Delta$ and $\Lambda^k$'s satisfy the commutation relation 
\beq
  \Lambda^kQ^\Delta = Q^k Q^\Delta\Lambda^k. 
\label{Lambda-QD-relation}
\eeq
The order of $\Gamma_{\pm},\Gamma'_{\pm}$ and $Q^\Delta$'s 
in (\ref{gencon-U}) can be thereby exchanged as 
\beq
\begin{aligned}
  \Gamma_{\pm}(q^{-\rho})Q^\Delta 
    &= Q^\Delta\Gamma_{\pm}(Q^{\pm 1}q^{-\rho}),\\
  \Gamma'_{\pm}(q^{-\rho})Q^\Delta 
    &= Q^\Delta\Gamma'_{\pm}(Q^{\pm 1}q^{-\rho}). 
\end{aligned}
\label{Gamma-QD-relation}
\eeq
\end{remark}

\begin{remark}
(\ref{Gamma(q)-matrix}) can be slightly generalized as 
\beq
\begin{aligned}
  \Gamma_{\pm}(Qq^{-\rho}) 
    &= \prod_{i=1}^\infty (1 - Qq^{i-1/2}\Lambda^{\pm 1})^{-1},\\
  \Gamma'_{\pm}(Qq^{-\rho}) 
    &= \prod_{i=1}^\infty (1 + Qq^{i-1/2}\Lambda^{\pm 1}). 
\end{aligned}
\label{Gamma(Q,q)-matrix}
\eeq
\end{remark}

\begin{remark}
$q^{\Delta^2/2}$ transforms $\Lambda^k$ by conjugation as 
\beq
  q^{\Delta^2/2}\Lambda^kq^{-\Delta^2/2} = q^{-k\Delta-k^2/2}\Lambda^k. 
\label{qD2-Lambda-relation}
\eeq
This is a special case of ``the third shift symmetry'' 
in the terminology of our previous work \cite{Takasaki13}.  
$\Lambda^k$ and $q^{-k\Delta-k^2/2}\Lambda^k$ 
are elements of the quantum torus algebra spanned by 
\beq
  v^{(k)}_m = q^{km/2}q^{k\Delta}\Lambda^m
                   = q^{-km/2}\Lambda^mq^{k\Delta}, \quad 
  k, m \in \ZZ. 
\label{q-torus-algebra}
\eeq
\end{remark}

\section{Initial values of dressing operators}

Given the infinite matrix $U$ representing an element $g$ 
of $\widehat{\GL}(\infty)$ 
\footnote{To give a solution of the Toda hierarchy, 
$g$ has to satisfy a set of genericity conditions, 
e.g., $\langle s|g|s\rangle \not= 0$ for all $s\in\ZZ$ \cite{Takebe91}. }, 
the associated solution of the Toda hierarchy can be characterized 
by the factorization problem \cite{Takasaki84,NTT95,Takasaki95}
\beq
  \exp\left(\sum_{k=1}^\infty t_k\Lambda^k\right) 
  U\exp\left(- \sum_{k=1}^\infty\bar{t}_k\Lambda^{-k}\right) 
  = W^{-1}\bar{W}, 
\label{Uttbar-problem}
\eeq
where $W$ is a lower triangular matrix 
with all diagonal elements being equal to $1$, 
and $\bar{W}$ is an upper triangular matrix 
with all diagonal elements being non-zero.  

$W$ and $\bar{W}$ play the role of dressing operators.   
As matrix-valued functions of $\bst$ and $\bar{\bst}$, 
they satisfy the Sato equations 
\beq
  \begin{aligned}
  &\frac{\rd W}{\rd t_k} 
     = - \left(W\Lambda^kW^{-1}\right)_{<0}W, &
  &\frac{\rd W}{\rd\bar{t}_k} 
     = \left(\bar{W}\Lambda^{-k}\bar{W}^{-1}\right)_{<0}W,\\
  &\frac{\rd \bar{W}}{\rd t_k} 
     = \left(W\Lambda^kW^{-1}\right)_{\ge 0}\bar{W},&
  &\frac{\rd\bar{W}}{\rd\bar{t}_k} 
     = - \left(\bar{W}\Lambda^{-k}\bar{W}^{-1}\right)_{\ge 0}\bar{W},
  \end{aligned}
\eeq
where $(\quad)_{\ge 0}$ and $(\quad)_{<0}$ denote 
the projection to the upper and strictly lower triangular parts, i.e., 
\beqnn
  \left(\sum_{m,n}a_{mn}E_{mn}\right)_{\ge 0} = \sum_{m\le n}a_{mn}E_{mn},\; 
  \left(\sum_{m,n}a_{mn}E_{mn}\right)_{<0} = \sum_{m>n}a_{mn}E_{mn}. 
\eeqnn
The Lax operators are obtained by ``dressing'' 
the shift matrix $\Lambda$  as 
\beqnn
  L = W\Lambda W^{-1}, \quad 
  \bar{L} = \bar{W}\Lambda\bar{W},  
\eeqnn
and satisfy the Lax equations 
\beq
  \begin{aligned}
  &\frac{\rd L}{\rd t_k} = [(L^k)_{\ge 0},\,L], & 
  &\frac{\rd L}{\rd\bar{t}_k} = [(\bar{L}^{-k})_{<0},\,L]\\
  &\frac{\rd\bar{L}}{\rd t_k} = [(L^k)_{\ge 0},\,\bar{L}],&
  &\frac{\rd\bar{L}}{\rd\bar{t}_k} = [(\bar{L}^{-k})_{<0},\,\bar{L}]. 
  \end{aligned}
\label{Toda-Lax-eq}
\eeq
Note that the inverse matrix $\bar{L}^{-1}$, 
too, satisfies Lax equations of the same form 
\beqnn
  \frac{\rd\bar{L}^{-1}}{\rd t_k} = [(L^k)_{\ge 0},\,\bar{L}^{-1}],\quad
  \frac{\rd\bar{L}^{-1}}{\rd\bar{t}_k} = [(\bar{L}^{-k})_{<0},\,\bar{L}^{-1}]. 
\eeqnn
Actually, it is $\bar{L}^{-1}$ rather than $\bar{L}$ 
that plays a fundamental role in our consideration 
on the solution obtained from (\ref{gencon-U}).  

In a general case, solving the factorization problem (\ref{Uttbar-problem}) 
directly is extremely difficult.  In the case of (\ref{gencon-U}), 
however,  we can find an explicit form of the solution 
at the initial time $\bst = \bar{\bst} = \bszero$ as follows. 

\begin{proposition}
The matrix (\ref{gencon-U}) can be factorized as 
\beq
  U = W_{(0)}^{-1}\bar{W}_{(0)},
\eeq
where 
\begin{align}
  W_{(0)} &= q^{\Delta^2/2}\cdot\prod_{n=1}^N
            \Gamma_{-}(Q^{(2n-1)}q^{-\rho})^{-1}
            \Gamma'_{-}(Q^{(2n)}q^{-\rho})^{-1}
          \cdot q^{-\Delta^2/2},
  \label{gencon-W}\\
  \bar{W}_{(0)} 
   &= q^{\Delta^2/2}\cdot\prod_{n=1}^N 
           \Gamma_{+}(Q^{(2n-1)-1}q^{-\rho})
           \Gamma'_{+}(Q^{(2n)-1}q^{-\rho}) \nonumber\\
   &\quad\mbox{}\times
      (Q_1\cdots Q_{2N-1})^\Delta q^{-\Delta^2/2}, 
   \label{gencon-Wbar}
\end{align}
and 
\beqnn
  Q^{(n)} = Q_1\cdots Q_{n-1},\quad Q^{(1)} = 1. 
\eeqnn
The factors $W_{(0)},\bar{W}_{(0)}$ are the initial values 
$W|_{\bst=\bar{\bst}=\bszero},\,\bar{W}|_{\bst=\bar{\bst}=\bszero}$ 
of the dressing matrices for the solution of the Toda hierarchy 
obtained from (\ref{gencon-U}). 
\end{proposition}

\begin{proof}
Use the commutation relation (\ref{Gamma-QD-relation}) 
repeatedly to move $Q_1^\Delta$, \ldots, $Q_{2N-1}^\Delta$ 
in (\ref{gencon-U}) to the front of $q^{-\Delta^2/2}$ as 
\begin{align*}
  U &= q^{\Delta^2/2}\Gamma_{-}(q^{-\rho})\Gamma_{+}(q^{-\rho})
       \Gamma'_{-}(Q_1q^{-\rho}))\Gamma'_{+}(Q_1^{-1}q^{-\rho})\nonumber\\
    &\quad\mbox{}\times
       \Gamma_{-}(Q^{(3)}q^{-\rho})\Gamma_{+}(Q^{(3)-1}q^{-\rho})
       \Gamma'_{-}(Q^{(4)}q^{-\rho})\Gamma'_{+}(Q^{(4)-1}q^{-\rho})\nonumber\\
    &\quad\mbox{}\times\cdots\nonumber\\
    &\quad\mbox{}\times
       \Gamma_{-}(Q^{(2N-1)}q^{-\rho})\Gamma_{+}(Q^{(2N-1)-1}q^{-\rho})
       \Gamma'_{-}(Q^{(2N)}q^{-\rho})\Gamma'_{+}(Q^{(2N)-1}q^{-\rho})\nonumber\\
    &\quad\mbox{}\times (Q_1\cdots Q_{2N-1})^{\Delta}q^{-\Delta^2/2}. 
\end{align*}
Since the vertex operators in the matrix representation 
(\ref{Gamma(z)-matrix}) commute with each other, 
$\Gamma_{-}$'s and $\Gamma'_{-}$'s can be moved 
to the left side of $\Gamma_{+}$'s and $\Gamma'_{+}$'s.   
Thus $U$ can be converted to the factorized form 
with factors (\ref{gencon-W}) and (\ref{gencon-Wbar}).  
Since (\ref{gencon-W}) is a lower triangular matrix 
with unit diagonal elements 
and (\ref{gencon-Wbar}) is an upper triangular matrix 
with non-zero diagonal elements, these matrices are indeed 
the initial values of the dressing operators. 
\end{proof}

\section{Initial values of Lax operators}

Let us calculate the initial values 
\beq
  L_{(0)} = W_{(0)}\Lambda W_{(0)}^{-1},\quad 
  \bar{L}_{(0)}^{-1} = \bar{W}_{(0)}\Lambda^{-1}\bar{W}_{(0)}^{-1}
\label{LLbar-5.1}
\eeq
of the Lax operators at $\bst = \bar{\bst} = \bszero$ 
from the explicit expressions 
(\ref{gencon-W}) and (\ref{gencon-Wbar})  
of the initial values of the dressing operators.  
As it turns out, they are non-commutative rational functions 
of $q^\Delta$ and $\Lambda$.  

Plugging the explicit form of $W_{(0)}$ and $\bar{W}_{(0)}$ 
into (\ref{LLbar-5.1}) yields the following 
somewhat complicated expressions of $L_{(0)}$ and $\bar{L}_{(0)}^{-1}$: 
\begin{align*}
  L_{(0)} &= q^{\Delta^2/2}\prod_{n=1}^N
       \Gamma_{-}(Q^{(2n-1)}q^{-\rho})^{-1}
       \Gamma'_{-}(Q^{(2n)}q^{-\rho})^{-1}\nonumber\\
    &\quad\mbox{}\times q^{-\Delta^2/2}\Lambda q^{\Delta^2/2}\nonumber\\
    &\quad\mbox{}\times\prod_{n=1}^N
       \Gamma'_{-}(Q^{(2n)}q^{-\rho})\Gamma_{-}(Q^{(2n-1)}q^{-\rho})
       \cdot q^{-\Delta^2/2}, 
\end{align*}
\begin{align*}
  \bar{L}_{(0)}^{-1}  & = q^{\Delta^2/2}\prod_{n=1}^N 
       \Gamma_{+}(Q^{(2n-1)-1}q^{-\rho})
       \Gamma'_{+}(Q^{(2n)-1}q^{-\rho}) \nonumber\\
    &\quad\mbox{}\times(Q_1\cdots Q_{2N-1})^\Delta 
      q^{-\Delta^2/2}\Lambda^{-1}
      q^{\Delta^2/2}(Q_1\cdots Q_{2N-1})^{-\Delta}\nonumber\\
    &\quad\mbox{}\times\prod_{n=1}^N
      \Gamma'_{+}(Q^{(2n)-1}q^{-\rho}) ^{-1}
      \Gamma_{+}(Q^{(2n-1)-1}q^{-\rho})^{-1}
      \cdot q^{-\Delta^2/2}. 
\end{align*}
We can use (\ref{qD2-Lambda-relation}) and 
(\ref{Lambda-QD-relation}) to rewrite 
the products in the middle of the right hand side as 
\beqnn
\begin{gathered}
   q^{-\Delta^2/2}\Lambda q^{-\Delta^2/2} = q^{\Delta+1/2}\Lambda,\\
   (Q_1\cdots Q_{2N-1})^\Delta q^{-\Delta^2/2}\Lambda^{-1}
     q^{\Delta^2/2}(Q_1\cdots Q_{2N-1})^{-\Delta}    
   = Q_1\cdots Q_{2N-1} q^{-\Delta+1/2}\Lambda^{-1}. 
\end{gathered}
\eeqnn
Since the last factors $\Lambda^{\pm 1}$ in the right hand side 
commute with $\Gamma$'s and $\Gamma'$'s,  
let us move them to the front of $q^{-\Delta^2/2}$ as 
\begin{align}
  L_{(0)} &= q^{\Delta^2/2}\prod_{n=1}^N
        \Gamma_{-}(Q^{(2n-1)}q^{-\rho})^{-1}
        \Gamma'_{-}(Q^{(2n)}q^{-\rho})^{-1}
       \cdot q^{\Delta+1/2}\nonumber\\
    &\quad\mbox{}\times\prod_{n=1}^N
        \Gamma'_{-}(Q^{(2n)}q^{-\rho})
        \Gamma_{-}(Q^{(2n-1)}q^{-\rho})
       \cdot\Lambda q^{-\Delta^2/2} 
\label{L-5.2}
\end{align}
and 
\begin{align}
  \bar{L}_{(0)}^{-1} 
    & = Q_1\cdots Q_{2N-1}q^{\Delta^2/2}
        \prod_{n=1}^N 
         \Gamma_{+}(Q^{(2n-1)-1}q^{-\rho})
         \Gamma'_{+}(Q^{(2n)-1}q^{-\rho}) 
        \cdot q^{-\Delta+1/2}\nonumber\\
    &\quad\mbox{}\times \prod_{n=1}^N
        \Gamma'_{+}(Q^{(2n)-1}q^{-\rho}) ^{-1}
        \Gamma_{+}(Q^{(2n-1)-1}q^{-\rho})^{-1}
       \cdot\Lambda^{-1} q^{-\Delta^2/2}. 
\label{Lbar-5.2}
\end{align}

To simplify these expressions, we use the following identities 
that underlie the ``shift symmetries'' of the quantum torus algebra 
\cite{NT07,NT08,Takasaki13}.  

\begin{lemma}
\beq
\begin{aligned}
  \Gamma_{-}(Qq^{-\rho})^{-1}q^\Delta\Gamma_{-}(Qq^{-\rho})
    &= q^\Delta(1 - Qq^{-1/2}\Lambda^{-1}),\\
  \Gamma'_{-}(Qq^{-\rho})^{-1}q^\Delta\Gamma'_{-}(Qq^{-\rho}) 
    &= q^\Delta(1 + Qq^{-1/2}\Lambda^{-1})^{-1},\\
  \Gamma_{+}(Q^{-1}q^{-\rho})q^{-\Delta}\Gamma_{+}(Q^{-1}q^{-\rho})^{-1} 
    &= q^{-\Delta}(1 - Q^{-1}q^{-1/2}\Lambda)^{-1},\\
  \Gamma'_{+}(Q^{-1}q^{-\rho})q^{-\Delta}\Gamma'_{+}(Q^{-1}q^{-\rho})^{-1}
    &= q^{-\Delta}(1 + Q^{-1}q^{-1/2}\Lambda). 
\end{aligned}
\label{Gamma-qD-relation}
\eeq
\end{lemma}

\begin{proof}
Use the commutation relations 
\beq
  \Lambda^kq^\Delta = q^kq^\Delta\Lambda^k 
\label{Lambda-qD-relation}
\eeq
(which is part of the structure of the quantum torus algebra) 
to exchange the order of $\Gamma_{-}(Qq^{-\rho})$ and $q^\Delta$ as 
\beqnn
  \Gamma_{-}(Qq^{-\rho})^{-1}q^\Delta\Gamma_{-}(Qq^{-\rho})
  = q^\Delta\Gamma_{-}(Qq^{-1}q^{-\rho})^{-1}\Gamma_{-}(Qq^{-\rho}), 
\eeqnn
and plug the infinite product expression 
(\ref{Gamma(Q,q)-matrix}) herein.  The outcome is 
\beqnn
  q^\Delta\prod_{i=1}^\infty(1 - Qq^{i-3/2}\Lambda^{-1})
  \cdot\prod_{i=1}^\infty(1 - Qq^{i-1/2}\Lambda^{-1})^{-1} 
  = q^\Delta(1 - Qq^{-1/2}\Lambda^{-1}). 
\eeqnn
Thus the first identity of (\ref{Gamma-qD-relation}) is obtained.  
The other identities can be derived in the same way.  
\end{proof}

The foregoing expressions (\ref{L-5.2})--(\ref{Lbar-5.2}) 
of $L_{(0)}$ and $\bar{L}_{(0)}^{-1}$ can be simplified 
by the identities (\ref{Gamma-qD-relation}) as 
\begin{align}
  L_{(0)} &= q^{\Delta^2/2}q^{\Delta+1/2} 
      \prod_{n=1}^N
       (1 - Q^{(2n-1)}q^{-1/2}\Lambda^{-1})
       (1 + Q^{(2n)}q^{-1/2}\Lambda^{-1})^{-1}
      \cdot\Lambda q^{-\Delta^2/2}
\label{L-5.3}
\end{align}
and 
\begin{align}
  \bar{L}_{(0)}^{-1} 
  & = Q_1\cdots Q_{2N-1}q^{\Delta^2/2}q^{-\Delta+1/2}\nonumber\\
  &\quad\mbox{}\times
      \prod_{n=1}^N 
        (1 - Q^{(2n-1)-1}q^{-1/2}\Lambda)^{-1}
        (1 + Q^{(2n)-1}q^{-1/2}\Lambda) 
      \cdot\Lambda^{-1} q^{-\Delta^2/2}. 
\label{Lbar-5.3}
\end{align}

Lastly, we use (\ref{qD2-Lambda-relation}) and 
(\ref{Lambda-qD-relation}) to convert (\ref{L-5.3}) and (\ref{Lbar-5.3}) 
to rational functions of $q^\Delta$ and $\Lambda$ as 
\begin{align}
  L_{(0)} &= q^{\Delta+1/2}\prod_{n=1}^N
         (1 - Q^{(2n-1)}q^{\Delta-1}\Lambda^{-1})
         (1 + Q^{(2n)}q^{\Delta-1}\Lambda^{-1})^{-1}
       \cdot q^{-\Delta-1/2}\Lambda \nonumber\\
  &= \prod_{n=1}^N
         (1 - Q^{(2n-1)}q^\Delta\Lambda^{-1})
         (1 + Q^{(2n)}q^\Delta\Lambda^{-1})^{-1}
       \cdot\Lambda 
  \label{L-5.4}
\end{align}
and 
\begin{align}
  \bar{L}_{(0)}^{-1} 
  &= Q_1\cdots Q_{2N-1}q^{-\Delta+1/2}\nonumber\\
  &\quad\mbox{}\times 
     \prod_{n=1}^N 
       (1 - Q^{(2n-1)-1}q^{-\Delta-1}\Lambda)^{-1}
       (1 + Q^{(2n)-1}q^{-\Delta-1}\Lambda) 
     \cdot q^{\Delta-1/2}\Lambda^{-1}\nonumber\\
  &= Q_1\cdots Q_{2N-1}
     \prod_{n=1}^N 
       (1 - Q^{(2n-1)-1}q^{-\Delta}\Lambda)^{-1}
       (1 + Q^{(2n)-1}q^{-\Delta}\Lambda) 
     \cdot\Lambda^{-1}. 
  \label{Lbar-5.4}
\end{align}

Although this is an almost final form, it will be suggestive 
to rewrite these expressions to the following ``symmetric'' form 
with the aid of (\ref{Lambda-qD-relation}).  

\begin{proposition}
The initial values (\ref{LLbar-5.1}) of the Lax operators 
have the factorized form 
\begin{align}
  L_{(0)} &= \prod_{n=1}^N(1 - Q^{(2n-1)}q^\Delta\Lambda^{-1})
        \cdot\Lambda\cdot
        \prod_{n=1}(1 + Q^{(2n)}\Lambda^{-1}q^\Delta)^{-1},
  \label{L-5.5}\\
  \bar{L}_{(0)}^{-1} 
    &= Q_1\cdots Q_{2N-1} 
        \prod_{n=1}^N(1 + Q^{(2n)-1}q^{-\Delta}\Lambda)
        \cdot\Lambda^{-1}\cdot 
        \prod_{n=1}^N(1 - Q^{(2n-1)-1}\Lambda q^{-\Delta})^{-1}. 
  \label{Lbar-5.5}
\end{align}
\end{proposition}

To make contact with a generalization of the Ablowitz-Ladik hierarchy, 
let us restate this result in a slightly different form.  
The fully factorized form (\ref{L-5.5}) of $L_{(0)}$ 
can be reassembled as 
\beq
  L_{(0)} = B_{(0)}\Lambda^{1-N}C_{(0)}^{-1}, 
\label{L-5.6}
\eeq
where 
\beq
\begin{aligned}
  B_{(0)} 
    &= \prod_{n=1}^N(1 - Q^{(2n-1)}q^\Delta\Lambda^{-1})\cdot\Lambda^N,\\
  C_{(0)} 
    &= \prod_{n=1}^N(1 + Q^{(2n)}\Lambda^{-1}q^\Delta). 
\end{aligned}
\label{gencon-BC}
\eeq
Moreover, since 
\beqnn
\begin{gathered}
  \prod_{n=1}^N(1 + Q^{(2n)-1}q^{-\Delta}\Lambda) 
    = Q^{(2)-1}Q^{(4)-1}\cdots Q^{(2N)-1}C(q^{-\Delta}\Lambda)^N,\\
  \prod_{n=1}^N(1 - Q^{(2n-1)-1}\Lambda q^{-\Delta})^{-1}
    = Q^{(1)}Q^{(3)}\cdots Q^{(2N-1)}(-\Lambda q^{-\Delta})^{-N}\Lambda^NB^{-1}
\end{gathered}
\eeqnn
and 
\beqnn
  (q^{-\Delta}\Lambda)^N\Lambda^{-1}(-\Lambda q^{-\Delta})^{-N}\Lambda^N
  = (-1)^N\Lambda^{N-1}, 
\eeqnn
the fully factorized form (\ref{Lbar-5.5}) of $\bar{L}_{(0)}^{-1}$,  
too, can be cast into the similar expression 
\beq
  \bar{L}_{(0)}^{-1} = DC_{(0)}\Lambda^{N-1}B_{(0)}^{-1}, 
\label{Lbar-5.6}
\eeq
where 
\beq
  D = (-1)^NQ_1\cdots Q_{2N-1}
         \frac{Q^{(1)}Q^{(3)}\cdots Q^{(2N-1)}}{Q^{(2)}Q^{(4)}\cdots Q^{(2N)}}
     = (-1)^NQ_2Q_4\cdots Q_{2N-2}. 
\label{gencon-D}
\eeq

$B_{(0)}$ and $C_{(0)}$  are polynomials in $\Lambda$ and $\Lambda^{-1}$, 
respectively, of the form 
\beqnn
\begin{aligned}
  B_{(0)} &= \Lambda^N + b_1\Lambda^{N-1} + \cdots + b_N, \\
  C_{(0)} &= 1 + c_1\Lambda^{-1} + \cdots + c_N\Lambda^{-N}, 
\end{aligned}
\eeqnn
where $b_1,\ldots,b_N$ and $c_1,\ldots,c_N$ 
are diagonal matrices. These matrices correspond 
to finite order difference operators of the form 
(\ref{general-BCop}). As we show in the next section, 
the partially factorized forms (\ref{L-5.6}) and (\ref{Lbar-5.6}) 
---  equivalently (\ref{LLbar-BCop}) as difference operators  --- 
of $L$ and $\bar{L}^{-1}$ yield a reduction of the Toda hierarchy 
to the generalized Ablowitz-Ladik hierarchy.   In this sense, 
(\ref{gencon-BC}) are initial conditions for a special solution 
of the generalized Ablowitz-Ladik hierarchy.  

\begin{remark}
The inverse matrices $B_{(0)}^{-1}$ and $C_{(0)}^{-1}$ are understood 
to be power series of $\Lambda$ and $\Lambda^{-1}$, respectively: 
\beqnn
\begin{aligned}
  B_{(0)}^{-1} &= b_N^{-1} - b_N^{-1}b_{N-1}\Lambda b_N^{-1} + \cdots,\\
  C_{(0)}^{-1} &= 1 - c_1\Lambda^{-1} + \cdots. 
\end{aligned}
\eeqnn
Since $B_{(0)}^{-1}$ and $C_{(0)}^{-1}$ are thus interpreted, 
the inverse of the right hand side of (\ref{Lbar-5.6}) 
is NOT equal to $D^{-1}B_{(0)}\Lambda^{1-N}C_{(0)}^{-1} = D^{-1}L_{(0)}$. 
Thus (\ref{L-5.6}) and (\ref{Lbar-5.6}) do not imply 
the trivial situation where $\bar{L}_{(0)}$ is proportional 
to $L_{(0)}$.  
\end{remark}

\section{Generalized Ablowitz-Ladik hierarchy}

The generalized Ablowitz-Ladik hierarchy is defined 
by the modified Lax equations 
\beq
\begin{aligned}
  \frac{\rd B}{\rd t_k}  &= P_kB - BQ_k,&
  \frac{\rd B}{\rd\bar{t}_k} &= \bar{P}_kB - B\bar{Q}_k,\\
  \frac{\rd C}{\rd t_k} &= P_kC - CR_k,&
  \frac{\rd C}{\rd\bar{t}_k} &= \bar{P}_kC - C\bar{R}_k 
\end{aligned}
\label{BC-eq}
\eeq
for general matrices $B,C$ of the form 
\beq
\begin{aligned}
  B &= \Lambda^N + b_1\Lambda^{N-1} + \cdots + b_N, \\
  C &= 1 + c_1\Lambda^{-1} + \cdots + c_N\Lambda^{-N}, 
\end{aligned}
\label{general-BC}
\eeq
where $b_1,\ldots,b_N$ and $c_1,\ldots,c_N$ are diagonal matrices.  
$P_k,Q_k,R_k$ and $\bar{P}_k,\bar{Q}_k,\bar{R}_k$ are defined as 
\beqnn
\begin{aligned}
  P_k &= \left((B\Lambda^{1-N}C^{-1})^k\right)_{\ge 0} ,&
  \bar{P}_k &= \left((DC\Lambda^{N-1}B^{-1})^k\right)_{<0},\\
  Q_k &= \left((\Lambda^{1-N}C^{-1}B)^k\right)_{\ge 0},&
  \bar{Q}_k &= \left((DB^{-1}C\Lambda^{N-1})^k\right)_{<0},\\ 
  R_k &= \left((C^{-1}B\Lambda^{1-N})^k\right)_{\ge 0},&
  \bar{R}_k &= \left((D\Lambda^{N-1}B^{-1}C)^k\right)_{<0},
\end{aligned}
\eeqnn
and $D$ is an arbitrary constant.   The case of $N =1$ 
amounts to the ordinary Ablowitz-Ladik hierarchy \cite{BCR11}. 

\begin{remark}
$D$ may be absorbed into rescaling of the time variables 
$\bar{t}_k \to D^{-k}\bar{t}_k$.  
\end{remark}

\begin{remark}
$Q_k,R_k,\bar{Q}_k,\bar{R}_k$ satisfy the algebraic relations 
\beq
  Q_k\Lambda^{1-N} - \Lambda^{1-N}R_k = 0, \quad
  \Lambda^{N-1}\bar{Q}_k - \bar{R}_k\Lambda^{N-1}  = 0. 
\label{QR-relation}
\eeq
These relations may be thought of as part of the following 
``cyclic'' reformulation of (\ref{BC-eq}): 
\beq
\begin{aligned}
  \frac{\rd B}{\rd t_k}  &= P_kB - BQ_k,&
  \frac{\rd C}{\rd\bar{t}_k} &= \bar{P}_kC - C\bar{R}_k,\\
  \frac{\rd\Lambda^{1-N}}{\rd t_k} 
     & = Q_k\Lambda^{1-N} - \Lambda^{1-N}R_k,&
  \frac{\rd\Lambda^{N-1}}{\rd\bar{t}_k} 
     &= \bar{R}_k\Lambda^{N-1} - \Lambda^{N-1}\bar{Q}_k,\\
    \frac{\rd C^{-1}}{\rd t_k} &= R_kC^{-1} - C^{-1}P_k,&
  \frac{\rd B^{-1}}{\rd\bar{t}_k} &=  \bar{Q}_kB^{-1} - B^{-1}\bar{P}_k. 
\end{aligned}
\label{cyclic-BC-eq}
\eeq
\end{remark}

Our goal is to verify that (\ref{BC-eq}) is a consistent reduction 
of the Toda hierarchy.   To avoid unnecessary complications 
of notations, we consider the case where $D = 1$. 
As remarked above, the general case can be reduced 
to this case by rescaling $\bar{t}_k$'s.  

Let us first confirm that the modified Lax equations (\ref{BC-eq}) 
imply the Lax equations (\ref{Toda-Lax-eq}) of the Toda hierarchy 
for matrices $L$ and $\bar{L}^{-1}$ defined 
as (\ref{L-5.6}) and (\ref{Lbar-5.6}) suggest.  

\begin{lemma}
\label{lemma:BC=>Toda-Lax}
If $B$ and $C$ satisfy (\ref{BC-eq}), then 
\beqnn
  L = B\Lambda^{1-N}C^{-1},\quad 
  \bar{L}^{-1} = C\Lambda^{N-1}B^{-1}
\eeqnn
satisfy the Lax equations (\ref{Toda-Lax-eq}) of the Toda hierarchy.  
\end{lemma}

\begin{proof}
By straightforward calculations, one can derive from (\ref{BC-eq}) 
and (\ref{QR-relation}) --- or, equivalently, from (\ref{cyclic-BC-eq}) --- 
the Lax equations 
\beqnn
\begin{aligned}
  \frac{\rd}{\rd t_k}(B\Lambda^{1-N}C^{-1}) 
    &= [P_k, B\Lambda^{1-N}C^{-1}],\\
  \frac{\rd}{\rd\bar{t}_k}(B\Lambda^{1-N}C^{-1}) 
    &= [\bar{P}_k, B\Lambda^{1-N}C^{-1}],\\
  \frac{\rd}{\rd t_k}(C\Lambda^{N-1}B^{-1}) 
    &= [P_k, C\Lambda^{N-1}B^{-1}],\\
  \frac{\rd}{\rd\bar{t}_k}(C\Lambda^{N-1}B^{-1}) 
    &= [\bar{P}_k, C\Lambda^{N-1}B^{-1}]. 
\end{aligned}
\eeqnn
Since 
\beqnn
  P_k = (L^k)_{\ge 0}, \quad 
  \bar{P}_k = (\bar{L}^{-k})_{<0}, 
\eeqnn
these equations coincide with the Lax equations 
of the Toda hierarchy.  
\end{proof}

We can further prove that (\ref{BC-eq}) can be converted to a system 
of evolution equations for the diagonal matrices $b_n,c_n$ 
or, equivalently, the functions $b_n(s),c_n(s)$ in (\ref{general-BCop}). 

\begin{lemma}
(\ref{BC-eq}) can be reduced to a system of evolution equations 
for $b_n,c_n$, $n = 1,\ldots,N$. 
\end{lemma}

\begin{proof}
The modified Lax equations (\ref{BC-eq}) have 
the equivalent (dual) form 
\beqnn
\begin{aligned}
  \frac{\rd B}{\rd t_k}  &= - P^c_kB + BQ^c_k, &
  \frac{\rd B}{\rd\bar{t}_k} &= -\bar{P}^c_kB + B\bar{Q}^c_k,\\
  \frac{\rd C}{\rd t_k} &= - P_k^cC + CR^c_k, &
  \frac{\rd C}{\rd\bar{t}_k} &= -\bar{P}^c_kC + C\bar{R}^c_k, 
\end{aligned}
\eeqnn
where 
\beqnn
\begin{aligned}
  P^c_k &= \left((B\Lambda^{1-N}C^{-1})^k\right)_{<0},&
  \bar{P}^c_k &= \left((C\Lambda^{N-1}B^{-1})^k)\right)_{\ge 0},\\
  Q^c_k &= \left((\Lambda^{1-N}C^{-1}B)^k\right)_{<0},&
  \bar{Q}^c_k &= \left((B^{-1}C\Lambda^{N-1})^k\right)_{\ge 0},\\
  R^c_k & = \left((C^{-1}B\Lambda^{1-N})^k\right)_{<0},&
  \bar{R}^c_k &= \left((\Lambda^{N-1}B^{-1}C)^k\right)_{\ge 0}. 
\end{aligned}
\eeqnn
These two expressions show that each equation of (\ref{BC-eq}) 
has such a form as 
\beqnn
\begin{aligned}
  &\frac{\rd B}{\rd t_k} = f_{1,k}\Lambda^{N-1} + \cdots + f_{N,k},& 
  &\frac{\rd B}{\rd\bar{t}_k} 
        = \bar{f}_{1,k}\Lambda^{N-1} + \cdots + \bar{f}_{N,k},\\
  &\frac{\rd C}{\rd t_k} = g_{1,k}\Lambda^{-1} + \cdots + g_{N,k}\Lambda^{-N},&
  &\frac{\rd C}{\rd\bar{t}_k} 
        = \bar{g}_{1,k}\Lambda^{-1} + \cdots + \bar{g}_{N,k}\Lambda^{-N},  
\end{aligned}
\eeqnn
hence can be reduced to the evolution equations 
\beq
  \frac{\rd b_n}{\rd t_k} = f_{n,k},\quad 
  \frac{\rd c_n}{\rd t_k} = g_{n,k},\quad 
  \frac{\rd b_n}{\rd\bar{t}_k} = \bar{f}_{n,k},\quad 
  \frac{\rd c_n}{\rd\bar{t}_k} = \bar{g}_{n,k} 
\label{bc-evol-eq}
\eeq
for $b_n$'s and $c_n$'s.  
\end{proof}

These two lemmas appear to be enough to conclude that 
(\ref{BC-eq}) defines a reduction of the Toda hierarchy.  
This is, however, a hasty judgement.   We have to confirm 
the consistency of (\ref{BC-eq}), namely, the commutativity 
\beq
\begin{gathered}
  \frac{\rd}{\rd t_k}\frac{\rd B}{\rd t_l} 
    = \frac{\rd}{\rd t_l}\frac{\rd B}{\rd t_k},\quad
  \frac{\rd}{\rd t_k}\frac{\rd B}{\rd\bar{t}_l} 
    = \frac{\rd}{\rd\bar{t}_l}\frac{\rd B}{\rd t_k},\quad
  \frac{\rd}{\rd\bar{t}_k}\frac{\rd B}{\rd\bar{t}_l} 
    = \frac{\rd}{\rd\bar{t}_l}\frac{\rd B}{\rd\bar{t}_k},\\
  \frac{\rd}{\rd t_k}\frac{\rd C}{\rd t_l} 
    = \frac{\rd}{\rd t_l}\frac{\rd C}{\rd t_k},\quad
  \frac{\rd}{\rd t_k}\frac{\rd C}{\rd\bar{t}_l} 
    = \frac{\rd}{\rd\bar{t}_l}\frac{\rd C}{\rd t_k},\quad
  \frac{\rd}{\rd\bar{t}_k}\frac{\rd C}{\rd\bar{t}_l} 
    = \frac{\rd}{\rd\bar{t}_l}\frac{\rd C}{\rd\bar{t}_k}
\end{gathered}
\label{comm-flow}
\eeq
of flows; otherwise, (\ref{BC-eq}) is not ensured 
to have non-trivial solutions.  This issue can be reduced 
to deriving a set of zero-curvature equations 
as shown below.  

\begin{lemma}
The three sets 
\beq
\begin{gathered}
  \frac{\rd P_l}{\rd t_k} - \frac{\rd P_k}{\rd t_l} + [P_l,P_k] = 0,\quad
  \frac{\rd\bar{P}_l}{\rd\bar{t}_k} - \frac{\rd\bar{P}_k}{\rd\bar{t}_l} 
    + [\bar{P}_l,\bar{P}_k] = 0,\\
  \frac{\rd\bar{P}_l}{\rd t_k} - \frac{\rd P_k}{\rd\bar{t}_l} 
    + [\bar{P}_l,P_k] = 0,
\end{gathered}
\label{PPbar-zc-eq}
\eeq
\beq
\begin{gathered}
  \frac{\rd Q_l}{\rd t_k} - \frac{\rd Q_k}{\rd t_l} + [Q_l,Q_k] = 0,\quad
  \frac{\rd\bar{Q}_l}{\rd\bar{t}_k} - \frac{\rd\bar{Q}_k}{\rd\bar{t}_l} 
    + [\bar{Q}_l,\bar{Q}_k] = 0,\\
  \frac{\rd\bar{Q}_l}{\rd t_k} - \frac{\rd Q_k}{\rd\bar{t}_l} 
    + [\bar{Q}_l,Q_k] = 0,
\end{gathered}
\label{QQbar-zc-eq}
\eeq
\beq
\begin{gathered}
  \frac{\rd R_l}{\rd t_k} - \frac{\rd R_k}{\rd t_l} + [R_l,R_k] = 0,\quad
  \frac{\rd\bar{R}_l}{\rd\bar{t}_k} - \frac{\rd\bar{R}_k}{\rd\bar{t}_l} 
    + [\bar{R}_l,\bar{R}_k] = 0,\\
  \frac{\rd\bar{R}_l}{\rd t_k} - \frac{\rd R_k}{\rd\bar{t}_l} 
    + [\bar{R}_l,R_k] = 0 
\end{gathered}
\label{RRbar-zc-eq}
\eeq
of zero-curvature equations imply the commutativity (\ref{comm-flow}) 
of flows.  
\end{lemma}

\begin{proof}
The difference of both hand sides of the first equation 
in (\ref{comm-flow}) can be rewritten as 
\begin{align*}
  \LHS - \RHS 
  &= \frac{\rd}{\rd t_k}(P_lB - BQ_l) 
     - \frac{\rd}{\rd t_l}(P_kB - BQ_k) \\
  &= \frac{\rd P_l}{\rd t_k}B + P_l(P_kB - BQ_k) 
     - (P_kB - BQ_k)Q_l - B\frac{\rd Q_l}{\rd t_k} \\
  &\quad\mbox{}
     - \frac{\rd P_k}{\rd t_k}B - P_k(P_lB - BQ_l) 
     + (P_lB - BQ_l)Q_k + B\frac{\rd Q_k}{\rd t_l}\\
  &= \left(\frac{\rd P_l}{\rd t_k} - \frac{\rd P_k}{\rd t_l} 
       + [P_l,P_k]\right)B 
     - B\left(\frac{\rd Q_l}{\rd t_k} - \frac{\rd Q_k}{\rd t_l} 
       + [Q_l,Q_k]\right). 
\end{align*}
This vanishes if the first equations of (\ref{PPbar-zc-eq}) 
and (\ref{QQbar-zc-eq}) are satisfied.  The other equations 
of (\ref{comm-flow}) can be derived in the same way.  
\end{proof}

The zero-curvature equations (\ref{PPbar-zc-eq})--(\ref{RRbar-zc-eq}) 
are indeed satisfied as we prove below.  

\begin{lemma}
(\ref{BC-eq}) implies the zero-curvature equations 
(\ref{PPbar-zc-eq})--(\ref{RRbar-zc-eq}). 
\end{lemma}

\begin{proof}
According to Lemma \ref{lemma:BC=>Toda-Lax}, 
\beqnn
  L = B\Lambda^{1-N}C^{-1},\quad 
  \bar{L}^{-1} = C\Lambda^{N-1}B^{-1}
\eeqnn
satisfy the Lax equations of the Toda hierarchy.    On the other hand, 
it is well known in the theory of the Toda hierarchy \cite{UT84} 
that the Lax equations imply the zero-curvature equations 
\beq
\begin{gathered}
  \frac{\rd(L^l)_{\ge 0}}{\rd t_k} - \frac{\rd(L^k)_{\ge 0}}{\rd t_l} 
  + [(L^l)_{\ge 0},(L^k)_{\ge 0}] = 0, \\
  \frac{\rd(\bar{L}^{-l})_{<0}}{\rd t_k} - \frac{\rd(\bar{L}^{-k})_{<0}}{\rd t_l} 
  + [(\bar{L}^{-l})_{<0},(\bar{L}^{-k})_{<0}] = 0, \\
  \frac{\rd(\bar{L}^{-l})_{<0}}{\rd t_k} 
     - \frac{\rd(L^k)_{\ge 0}}{\rd\bar{t}_l} 
    + [(\bar{L}^{-l})_{<0}, (L^k)_{\ge 0}] = 0. 
\end{gathered}
\label{Toda-zc-eq}
\eeq
Since $(L^k)_{\ge 0}$ and $(\bar{L}^{-k})_{<0}$ for $L$ and $\bar{L}$ 
of the foregoing form coincide with $P_k$ and $\bar{P}_k$, 
these zero-curvature equations are nothing but (\ref{PPbar-zc-eq}).   
In the same sense, both 
\beqnn
  L = \Lambda^{1-N}C^{-1}B, \quad 
  \bar{L}^{-1} = B^{-1}C\Lambda^{N-1}
\eeqnn
and 
\beqnn
  L = C^{-1}B\Lambda^{1-N},\quad 
  \bar{L}^{-1} = \Lambda^{N-1}B^{-1}C
\eeqnn
satisfy the Lax equations of the Toda hierarchy, 
and the associated zero-curvature equations 
are exactly (\ref{QQbar-zc-eq}) and (\ref{RRbar-zc-eq}).  
\end{proof}

We have thus confirmed the following fact.   
Note that the constant $D$ is restored in the statement.  

\begin{proposition}
(\ref{BC-eq}) defines a consistent system (generalized 
Ablowitz-Ladik hierarchy) of evolution equations 
for $b_n,c_n$, $n = 1,\ldots,N$.   This system can be identified 
with a subsystem of the Toda hierarchy for which 
the Lax operators $L,\bar{L}^{-1}$ have the factorized form 
\beq
  L = B\Lambda^{1-N}C^{-1},\quad 
  \bar{L}^{-1} = DC\Lambda^{N-1}B^{-1}. 
\label{LLbar-BC}
\eeq
\end{proposition}

This result ensures the existence of a solution of (\ref{BC-eq}) 
with the initial values (\ref{gencon-BC}) 
at $\bst = \bar{\bst} = \bszero$.   
The associated Lax operators of the Toda hierarchy 
have the initial values (\ref{L-5.5})--(\ref{Lbar-5.5}) 
at $\bst = \bar{\bst} = \bszero$.   
This leads to the following conclusion. 

\begin{proposition}
The generalized Ablowitz-Ladik hierarchy (\ref{BC-eq}) 
has a solution with the initial values (\ref{gencon-BC}) 
at $\bst = \bar{\bst} = \bszero$.    This solution corresponds 
to the solution of the Toda hierarchy obtained from (\ref{gencon-U}) 
by the factorization problem (\ref{Uttbar-problem}).  
\end{proposition}

\section{Conclusion}

We have thus proved that the generalized Ablowitz-Ladik hierarchy 
(\ref{BC-eq}) underlies open string amplitudes 
of topological string theory on the generalized conifolds.   
These amplitudes are assembled to the tau function 
$\tau(s,\bst,\bar{\bst})$ that is primarily formulated 
in the fermionic language.  To identify it as a solution 
of the generalized Ablowitz-Ladik hierarchy, we use 
the dictionary that translates the fermionic setting 
to the language of $\ZZ\times\ZZ$ matrices.   
The generating operator $g$ of the tau function is thereby 
converted to the matrix $U$ along with the building blocks 
such as the fermion bilinears $L_0,W_0$ 
and the vertex operators $\Gamma_{\pm}(z),\Gamma'_{\pm}(z)$.  
The subsequent consideration is based on these matrices, 
and can be summarized as follows.  

\begin{theorem}
The matrix $U$ determines a solution of the Toda hierarchy 
via the factorization problem (\ref{Uttbar-problem}).   
The associated Lax operators $L,\bar{L}^{-1}$ have 
the factorized form (\ref{LLbar-BC}).   The factors $B,C$ 
in (\ref{LLbar-BC}) are polynomials of the form (\ref{general-BC}) 
in the shift matrices $\Lambda,\Lambda^{-1}$ respectively, 
and satisfy the modified Lax equations (\ref{BC-eq}). 
Their initial values at $\bst = \bar{\bst} = \bszero$ 
and the constant $D$ are determined by the K\"ahler parameters 
as shown in (\ref{gencon-BC}) and (\ref{gencon-D}).  
\end{theorem}

A technical clue of the proof is the fact that 
the initial values (\ref{LLbar-5.1}) of the Lax operators 
at $\bst = \bar{\bst} = \bszero$ can be calculated 
in a very explicit form as shown in (\ref{L-5.2})--(\ref{Lbar-5.5}).   
It is ultimately the quantum torus algebra (\ref{q-torus-algebra}) 
and the matrix-valued quantum dilogarithm (\ref{Gamma(q)-matrix}) 
that enable us these calculations.  

Let us conclude this paper with several remarks.  

\paragraph{1.} 
The calculation of initial values can be extended 
to the Orlov-Schulman operators 
\beqnn
  M = W\Delta W^{-1} + \sum_{k=1}^\infty kt_kL^k,\quad 
  \bar{M} = \bar{W}\Delta\bar{W}^{-1} 
            - \sum_{k=1}^\infty k\bar{t}_k\bar{L}^{-k}. 
\eeqnn
In the context of the quantum torus algebra, 
it is more natural to consider the exponentials 
$q^M,q^{-\bar{M}}$ rather than $M,\bar{M}$ themselves.   
Their initial values $q^{M_{(0)}},q^{-\bar{M}_{(0)}}$ 
are determined by the initial values 
(\ref{gencon-W})--(\ref{gencon-Wbar}) 
of the dressing operators as 
\beq
  q^{M_{(0)}} = W_{(0)}q^\Delta W_{(0)}^{-1},\quad
  q^{-\bar{M}_{(0)}} = \bar{W}_{(0)}q^{-\Delta}\bar{W}_{(0)}^{-1}. 
\eeq
One can calculate the right hand side 
in much the same way as the derivation of (\ref{L-5.4}) 
and (\ref{Lbar-5.4}): 
\begin{align}
  q^{M_{(0)}} 
    &= \prod_{n=1}^N (1 - Q^{(2n-1)}q^\Delta\Lambda^{-1}) 
       (1 + Q^{(2n)}q^\Delta\Lambda^{-1})^{-1}\cdot q^\Delta,\\
  q^{-\bar{M}_{(0)}} 
    &= \prod_{n=1}^N (1 - Q^{(2n-1)-1}q^{-\Delta}\Lambda)^{-1} 
       (1 + Q^{(2n)-1}q^{-\Delta}\Lambda)\cdot q^{-\Delta}.     
\end{align}
These expressions of the initial values may be thought of 
as yet another realization of ``quantum mirror curves'' 
of the generalized conifolds \cite{GS11,Takasaki13b}.  
This issue will be addressed elsewhere. 

\paragraph{2.}
There are some other possible approaches 
to integrable structures of the type studied in this paper.  
An approach will be based on unitary matrix models 
of open string amplitudes \cite{Okuda04,ST10,OSY10}.  
As the ordinary unitary matrix model is linked 
with the Ablowitz-Ladik hierarchy \cite{KMZ96}, 
these matrix models with elliptic weight functions 
can lead to generalized Ablowitz-Ladik hierarchies.  
Another way is the line of approach initiated 
by Brini et al. \cite{Brini10,BCR11}.  
In the latest paper \cite{BCRR14}, 
they proposed a new scheme of reductions 
of the Toda hierarchy, which is apparently different 
from ours, and applied the reduced hierarchy 
to local Gromov-Witten invariants of a class 
of orbifolds of the resolved conifold.  
It is an extremely intriguing problem to compare 
their results with ours in detail.  
Since their toric Calabi-Yau threefolds, too, 
have web diagrams ``on a strip'', 
one can construct a Toda tau function 
from open string amplitudes in exactly the same manner 
as our case (see Appendix A).  

\paragraph{3.} 
It seems extremely difficult to go ``beyond the strip''. 
When the web diagram is not on a strip, 
the open string amplitudes obtained 
by the topological vertex construction cease 
to take the simple form $\langle\lambda|g|\mu\rangle$.  
A way out will be the B-model perspective 
\cite{ADKMV03,DHSV07,DHS08}, but the study of 
integrable structures on the B-model side 
require quite different techniques.  
A promising strategy on the A-model side will be 
the method of ``gluing strips'' proposed 
by Eynard et al. \cite{EKPM10}.

\subsection*{Acknowledgements}

This work is partly supported by JSPS KAKENHI Grant 
Numbers 24540223 and 25400111.

\appendix

\section{Combinatorial and fermionic expressions 
of open string amplitudes}

The method of topological vertex \cite{AKMV03} 
provides a combinatorial expression of 
open string amplitudes for toric Calabi-Yau threefolds. 
Fundamental building blocks of the construction 
are the vertex weights 
\footnote{The parameter $q$ in the early literature 
\cite{AKMV03} is replaced by $q^{-1}$ in order to 
fit the crystal model interpretation \cite{ORV03}.}
\beqnn
  C_{\alpha\beta\gamma} 
  = s_\beta(q^{-\rho})q^{-\kappa(\gamma)/2} 
    \sum_{\nu\in\calP}s_{\alpha/\nu}(q^{-\tp{\beta}-\rho})
      s_{\tp{\gamma}/\nu}(q^{-\beta-\rho}). 
\eeqnn
to be attached to the vertices of the web diagram.  
These vertex weights are ``glued together'' along with 
the weights $(-Q_n)^{|\alpha_n|}$ of the internal edges 
by summing over the partitions $\alpha_n$ assigned 
to the internal edges.  We refer details of 
the gluing procedure (in which extra ``framing factors'' 
$(-1)^{f_n}q^{f_n\kappa(\alpha_n)/2}$ are also inserted 
to adjust the canonical framing of the vertex weights) 
to Mari\~{n}o's book \cite{Marino-book}. 

To apply this method to the generalized conifolds, 
let $\alpha_0,\beta_1,\cdots,\beta_{2N},\alpha_{2N}$ 
denote the partitions assigned to the external edges 
of the web diagram (see Fig. \ref{web-diagram}).  
$\alpha_0$ and $\alpha_{2N}$ are on the leftmost 
and rightmost external edges; 
$\beta_1,\ldots,\beta_{2N}$ are on the intermediate ones. 
The amplitude $Z_{\alpha_0\beta_1\cdots\beta_{2N}\alpha_{2N}}$ 
in this setting becomes the following sum over all partitions 
$\alpha_1,\ldots,\alpha_{2N}$ on the internal edges: 
\begin{align}
  Z_{\alpha_0\beta_1\cdots\beta_{2N}\alpha_{2N}} 
  &= \sum_{\alpha_1,\ldots,\alpha_{2N}\in\calP}
     C_{\alpha_1\beta_1\alpha_0}(-Q_1)^{|\alpha_1|}
     C_{\tp{\alpha}_1\beta_2\alpha_2}(-Q_2)^{|\alpha_2|}\nonumber\\
  &\quad\mbox{}\times 
     C_{\alpha_3\beta_3\alpha_2}(-Q_3)^{|\alpha_3|}
     C_{\tp{\alpha}_3\beta_4\alpha_4}(-Q_4)^{|\alpha_4|}\nonumber\\
  &\quad\mbox{}\times\cdots \nonumber\\
  &\quad\mbox{}\times 
     C_{\alpha_{2N-1}\beta_{2N-1}\alpha_{2N-2}}(-Q_{2N-1})^{|\alpha_{2N-1}|}
     C_{\tp{\alpha}_{2N-1}\beta_{2N}\alpha_{2N}}. 
  \label{Z-vertex}
\end{align}

This combinatorial expression can be translated 
to the language of fermions as follows.  
The vertex weights have two equivalent fermionic expressions 
\begin{align}
  C_{\alpha\beta\gamma} 
  &= s_\beta(q^{-\beta})q^{-\kappa(\gamma)/2}
    \langle\tp{\gamma}|\Gamma_{-}(q^{-\beta-\rho})
    \Gamma_{+}(q^{-\tp{\beta}-\rho})|\alpha\rangle,
  \label{vertex-fermion}\\
  C_{\alpha\beta\gamma} 
  &= s_\beta(q^{-\beta})q^{-\kappa(\gamma)/2}
    \langle\tp{\alpha}|\Gamma'_{-}(q^{-\tp{\beta}-\rho})
    \Gamma'_{+}(q^{-\beta-\rho})|\gamma\rangle 
  \label{vertex-fermion'}
\end{align}
as (\ref{<ls|Gamma|ms>}) and (\ref{<ls|Gamma'|ms>}) imply.  
One can choose either (\ref{vertex-fermion}) 
or (\ref{vertex-fermion'}) appropriately in calculations. 
The edge weights, too, can be expressed as the matrix elements 
\beqnn
  (-Q_n)^{|\alpha_n|} = \langle\alpha_n|(-Q_n)^{L_0}|\alpha_n\rangle 
\eeqnn
of the ``diagonal'' operator $(-Q_n)^{L_0}$ by (\ref{<ls|L0|ms>}). 
When these expressions are plugged into (\ref{Z-vertex}), 
the somewhat troublesome factors $q^{\kappa(\alpha_n)/2}$ 
on the internal edges turn out to cancel out. 
(Actually, this is also the case for all web diagrams 
on a ``strip'' \cite{IKP04}.)   Summation over 
$\alpha_1,\ldots,\alpha_{2N-1}\in\calP$ can be interpreted 
as defining a multiple product of operators 
in the charge-0 sector of the fermion Fock space.  
Thus the amplitude becomes, apart from several prefactors, 
the matrix element of a single operator: 
\beq
  Z_{\alpha_0\beta_1\ldots\beta_{2N}\alpha_{2N}}
  = s_{\beta_1}(q^{-\rho})\cdots s_{\beta_{2N}}(q^{-\rho})
    q^{-(\kappa(\alpha_0)+\kappa(\alpha_{2N}))/2}
    \langle\tp{\alpha_0}|g_{\beta_1\ldots\beta_{2N}}|\alpha_{2N}\rangle, 
  \label{Z-fermion}
\eeq
where 
\begin{align}
  &g_{\beta_1\cdots\beta_{2N}} \nonumber\\
  &= \Gamma_{-}(q^{-\beta_1-\rho})\Gamma_{+}(q^{-\tp{\beta}_1-\rho})(-Q_1)^{L_0}
     \Gamma'_{-}(q^{-\tp{\beta_2}-\rho})\Gamma'_{+}(q^{-\beta_2-\rho})(-Q_2)^{L_0}
     \nonumber\\
  &\quad\mbox{}\times 
     \Gamma_{-}(q^{-\beta_3-\rho})\Gamma_{+}(q^{-\tp{\beta}_3-\rho})(-Q_3)^{L_0}
     \Gamma'_{-}(q^{-\tp{\beta}_4-\rho})\Gamma'_{+}(q^{-\beta_4-\rho})(-Q_4)^{L_0}
     \nonumber\\
  &\quad\mbox{}\times\cdots\nonumber\\
  &\quad\mbox{}\times 
     \Gamma_{-}(q^{-\beta_{2N-1}-\rho})\Gamma_{+}(q^{-\tp{\beta_}{2N-1}-\rho})
     (-Q_{2N-1})^{L_0}
    \Gamma'_{-}(q^{-\tp{\beta}_{2N}-\rho})\Gamma'_{+}(q^{-\beta_{2N}-\rho}). 
\end{align}

The operator $g$ of (\ref{gencon-g}) emerges by letting 
\beqnn
  \alpha_0 = \tp{\lambda},\quad 
  \beta_1 = \cdots = \beta_{2N} = \emptyset,\quad 
  \alpha_{2N} = \mu 
\eeqnn
and flipping the sign $Q_n \to -Q_n$ of all $Q_n$'s.  
More precisely, $g$ is related to the operator 
$g_{\emptyset\cdots\emptyset}$ as 
\beq
  g = q^{W_0/2}g_{\emptyset\cdots\emptyset}|_{Q_n\to -Q_n}q^{-W_0/2}. 
\eeq
Note here that 
$q^{\pm W_0/2}$ are ``diagonal'' operators with the matrix elements 
\beqnn
  \langle\lambda|q^{\pm W_0/2}|\mu\rangle 
  = \delta_{\lambda\mu}q^{\pm(\kappa(\lambda)+|\lambda|)/2} 
\eeqnn
as one can see from (\ref{<ls|L0|ms>}) and (\ref{<ls|W0|ms>}). 
Also recall the anti-symmetric property 
\beqnn
  \kappa(\tp{\lambda}) = - \kappa(\lambda)
\eeqnn
of the quadratic Casimir invariant. 
Consequently, the matrix elements of $g$ in the charge-0 sector 
and the open string amplitudes $Z_{\tp{\lambda}\emptyset\cdots\emptyset\mu}$
are related as 
\beq
  \langle\lambda|g|\mu\rangle 
  = q^{(|\lambda|-|\mu|)/2}
    Z_{\tp{\lambda}\emptyset\cdots\emptyset\mu}|_{Q_n\to -Q_n}. 
\eeq
Thus the essential features of the open string amplitudes 
are fully captured by the matrix elements of $g$.  
The prefactor $q^{(|\lambda|-|\mu|)/2}$ is harmless 
--- it can be absorbed into rescaling 
of the variables $\bsx,\bsy$ of the Schur functions 
in the generating function (\ref{gen-function}). 

The fermionic expression (\ref{Z-fermion}) of 
open string amplitudes can be generalized 
to other web diagrams on a strip \cite{EK03,Nagao09,Sulkowski09}.  
It will be obvious that those amplitudes, too, 
yield a tau function of the Toda hierarchy.


\begin{thebibliography}{99}

\bibitem{Takasaki13}
K.~Takasaki, 
Modified melting crystal model and Ablowitz-Ladik hierarchy
{\it J. Phys. A: Math. Theor.\/} {\bf 46} (2013), 245202, 
arXiv:1302.6129 [math-ph]. 

\bibitem{BP08}
J.~Bryan and R.~Pandharipande, 
The local Gromov-Witten theory of curves, 
J. Amer. Math. Soc. {\bf 21} (2008), 101--136, 
arXiv:math/0411037. 

\bibitem{AL75}
M.~J.~Ablowitz and J.~F.~Ladik, 
Nonlinear differential-difference equations, 
J. Math. Phys. {\bf 16} (1975), 598--603. 

\bibitem{Vekslerchik97}
V.~E.~Vekslerchik, 
Functional representation of the Ablowitz-Ladik hierarchy, 
J. Phys. A: Math. Gen. {\bf 31} (1998), 1087--1099, 
arXiv:solv-int/9707008. 

\bibitem{Ruijsenaars90}
S.~N.~M.~Ruijsenaars, 
Relativistic Toda systems, 
Comm. Math. Phys. {\bf 133} (1990), 217--247.

\bibitem{KMZ96}
S.~Kharchev, A.~Mironov and A.~Zhedanov, 
Faces of relativistic Toda chain, 
Int. J. Mod. Phys. {\bf A12} (1997), 2675--2724, 
arXiv:hep-th/9606144. 

\bibitem{Suris97}
Yu.~B.~Suris, 
A note on the integrable discretization 
of the nonlinear Schr\"odinger equation, 
Inverse Problems {\bf 13} (1997), 1121--1136, 
arXiv:solv-int/9701010. 

\bibitem{Brini10}
A.~Brini, The local Gromov-Witten theory of 
$\mathbb{CP}^1$ and integrable hierarchies, 
Comm. Math. Phys. {\bf 313} (2012), 571--605, 
arXiv:1002.0582 [math-ph].

\bibitem{NT07}
T.~Nakatsu and K.~Takasaki, 
Melting crystal, quantum torus and Toda hierarchy, 
Comm. Math. Phys. {\bf 285} (2009), 445--468, 
arXiv:0710.5339 [hep-th]. 

\bibitem{NT08}
T.~Nakatsu and K.~Takasaki, 
Integrable structure of melting crystal model with external potentials, 
Adv. Stud. Pure Math. vol. 59 (Mathematical Society of Japan, Tokyo, 2010), 
pp. 201--223, arXiv:0807.4970 [math-ph]. 

\bibitem{ORV03}
A.~Okounkov, N.~Reshetikhin and C.~Vafa,  
Quantum Calabi-Yau and classical crystals, 
P. Etingof, V. Retakh and I.M. Singer (eds.), 
{\it The unity of mathematics\/}, 
Progr. Math.  vol. 244 (Birkh\"auser, 2006), pp. 597--618, 
arXiv:hep-th/0309208

\bibitem{MNTT04}
T.~Maeda, T.~Nakatsu, K.~Takasaki and T.~Tamakoshi, 
Five-dimensional supersymmetric Yang-Mills theories 
and random plane partitions, 
JHEP {\bf 0503}  (2005), 056, 
arXiv:hep-th/0412327. 

\bibitem{UT84}
K.~Ueno and K.~Takasaki, 
Toda lattice hierarchy, 
Adv. Stud. Pure Math. vol. 4 (Kinokuniya, Tokyo, 1984), pp. 1--95. 

\bibitem{TT95}
K.~Takasaki and T.~Takebe, 
Integrable hierarchies and dispersionless limit, 
Rev. Math. Phys. {\bf 7} (1995), 743--808, 
arXiv:hep-th/9405096. 

\bibitem{BCR11}
A.~Brini, G.~Carlet and P.~Rossi, 
Integrable hierarchies and the mirror model of 
local $\mathbb{CP}^1$, 
Physica {\bf D241} (2012), 2156--2167. 
arXiv:1105.4508 [math.AG].

\bibitem{AKMV03}
M.~Aganagic, A.~Klemm, M.~Mari\~no and C.~Vafa,
The topological vertex, 
Comm. Math. Phys. {\bf 254} (2005), 425--478, 
arXiv:hep-th/0305132.

\bibitem{IKP04}
A.~Iqbal and A.-K. Kashani-Poor, 
The vertex on a strip, 
Adv. Theor. Math. Phys. {\bf 10} (2006), 317--343,  
arXiv:hep-th/0410174. 

\bibitem{EK03}
T.~Eguchi and H.~Kanno, 
Geometric transitions, Chern-Simons gauge theory 
and Veneziano type amplitudes, 
Phys. Lett. {\bf B585} (2004) 163--172, 
arXiv:hep-th/0312223. 

\bibitem{Nagao09}
K.~Nagao, 
Non-commutative Donaldson-Thomas theory 
and vertex operators, 
Geometry and Topology {\bf 15} (2011) 1509--1543, 
arXiv:0910.5477 [math.AG]. 

\bibitem{Sulkowski09}
P.~Su{\l}kowski, 
Wall-crossing, free fermions and crystal melting, 
Comm. Math. Phys. {\bf 301} (2011), 517--562, 
arXiv:0910.5485 [hep-th]. 

\bibitem{Takasaki13b}
K.~Takasaki, 
Remarks on partition functions of topological string theory 
on generalized conifolds, 
arXiv:1301.4548 [math-ph]. 

\bibitem{Takasaki84}
K.~Takasaki, 
Initial value problem for the Toda lattice hierarchy, 
Adv. Stud. Pure Math. vol. 4 (Kinokuniya, Tokyo, 1984), 
pp. 136--163. 

\bibitem{NTT95}
T.~Nakatsu, K.~Takasaki and S.~Tsujimaru, 
Quantum and classical aspects of deformed $c=1$ strings, 
Nucl. Phys. {\bf B443} (1995), 1550-197, 
arXiv:hep-th/9501038. 

\bibitem{Takasaki95}
K.~Takasaki, 
Toda lattice hierarchy and generalized string equations, 
Comm. Math. Phys. {\bf 181} (1996), 131--156, 
arXiv:hep-th/9506089. 

\bibitem{MJD-book}
T.~Miwa, M.~Jimbo and E.~Date, 
{\it Solitons: Differential equations, symmetries, 
and infinite-dimensional algebras\/}, 
Cambridge University Press, 2000. 

\bibitem{OR01}
A.~Okounkov and N.~Reshetikhin,  
Correlation function of Schur process with application 
to local geometry of a random 3-dimensional Young diagram,
J. Amer. Math. Soc. {\bf 16}, (2003), 581--603, 
arXiv:math.CO/0107056 

\bibitem{BY08}
J.~Bryan and B.~Young, 
Generating functions for coloured 3D Young diagrams 
and the Donaldson-Thomas invariants of orbifolds,
Duke Math. J. {\bf 152} (2010), 115--153, 
arXiv:0802.3948 [math.CO]. 

\bibitem{Macdonald-book} 
I.~G. Macdonald, 
{\it Symmetric functions and Hall polynomials\/}, 
Oxford University Press, 1995.

\bibitem{Takebe91}
T.~Takebe, 
Representation theoretical meanings of 
the initial value problem for 
the Toda lattice hierarchy I, 
Lett. Math. Phys. {\bf 21} (1991), 77--84. 

\bibitem{Marino-book}
M. Mari\~{n}o, 
{\it Chern-Simons theory, matrix models, and topological strings\/}, 
Oxford University Press, 2005.  

\bibitem{FV93}
L.~Faddeev and A.~Yu.~Volkov, 
Abelian current algebra and the Virasoro algebra on the lattice, 
Phys. Lett. {\bf B315} (1993), 311--318, 
arXiv:hep-th/9307048. 

\bibitem{FK93}
L.~D.~Faddeev and R.~M.~Kashaev, 
Quantum dilogarithm, 
Mod. Phys. Lett. {\bf A9} (1994), 427--434, 
arXiv:hep-th/9310070.

\bibitem{GS11}
S.~Gukov and P.~Su{\l}kowski, 
A-polynomial, B-model, and quantization, 
JHEP {\bf 1202} (2012), 070, arXiv:1108.0002v1 [hep-th]. 

\bibitem{Okuda04}
T.~Okuda, 
Derivation of Calabi-Yau crystals from Chern-Simons gauge theory,
JHEP {\bf 0503} (2005), 047, arXiv:hep-th/0409270. 

\bibitem{ST10}
R.~Szabo and M.~Tierz, 
Matrix models and stochastic growth 
in Donaldson-Thomas theory, 
J. Math. Phys. 53 (2012), 103502, arXiv:1005.5643 [hep-th].

\bibitem{OSY10}
H.~Ooguri, P.~Su{\l}kowski and M.~Yamazaki, 
Wall crossing as seen by matrix models, 
Commun. Math. Phys. {\bf 307} (2011), 429--462, 
arXiv:1005.1293 [hep-th]. 

\bibitem{BCRR14}
A.~Brini, G.~Carlet, S.~Romano and P.~Rossi, 
Rational reductions of the 2D-Toda hierarchy and mirror symmetry, 
arXiv:1401.5725 [math-ph]. 

\bibitem{ADKMV03}
M.~Aganagic, R.~Dijkgraaf, A.~Klemm, M.~Mari\~{n}o and C.~Vafa, 
Topological strings and integrable hierarchies, 
Commun. Math. Phys. {\bf 261} (2006), 451-516, 
arXiv:hep-th/0312085.

\bibitem{DHSV07}
R.~Dijkgraaf, L.~Hollands, P.~Su{\l}kowski and C.~Vafa, 
Supersymmetric gauge theories, intersecting branes and free fermions, 
JHEP {\bf 0802} (2008), 106, arXiv:0709.4446 [hep-th]. 

\bibitem{DHS08}
R.~Dijkgraaf, L.~Hollands and P.~Su{\l}kowski, 
Quantum curves and D-Modules, 
JHEP {\bf 0911} (2009), 047, arXiv:0810.4157 [hep-th]. 

\bibitem{EKPM10}
B.~Eynard, A.-K.~Kashani-Poor and O.~Marchal, 
A matrix model for the topological string I: 
Deriving the matrix model, 
arXiv:1003.1737 [hep-th]. 


\end{thebibliography}
\end{document}